\documentclass[12pt]{article}
\usepackage{amssymb,amsmath,natbib,graphicx,amsthm,
  setspace,sectsty,anysize,times,dsfont,enumerate}

\usepackage[svgnames]{xcolor}

\usepackage{lscape,arydshln,relsize,rotating,multirow}
\usepackage{caption}
\usepackage{algorithm,algorithmic}

\graphicspath{{graphs/}}

\newtheorem{theorem}{\sc Theorem}[section]
\newtheorem{definition}{\sc Definition}[section]
\newtheorem{lemma}{\sc Lemma}[section]

\marginsize{1.1in}{.9in}{.3in}{1.4in}

\newcommand{\dbl}{\setstretch{1.5}}
\newcommand{\sgl}{\setstretch{1.2}}

\newcommand{\bs}[1]{\boldsymbol{#1}}

\newcommand{\mr}[1]{\mathrm{#1}}
\newcommand{\bm}[1]{\mathbf{#1}}
\newcommand{\ds}[1]{\mathds{#1}}

\DeclareMathOperator*{\argmin}{argmin}

\usepackage[bottom,hang,flushmargin]{footmisc}

\begin{document}

\sgl 

\pagestyle{empty}

~
\vskip 3cm

\noindent {\LARGE \bf One-step estimator paths for concave regularization} 

\vskip 1cm

\noindent{\Large Matt Taddy}

{\large
\vskip .2cm \noindent
{Microsoft Research and the  University of Chicago Booth School of Business}\\
\texttt{faculty.chicagobooth.edu/matt.taddy}}

\vskip 2cm

{\noindent The statistics literature of the past 15 years has established many
favorable properties for sparse diminishing-bias regularization: techniques
which can roughly be understood as providing estimation under penalty
functions spanning the range of concavity between $\ell_0$ and $\ell_1$ norms.
However, lasso $\ell_1$-regularized estimation remains the standard tool for
industrial `Big Data' applications because of its minimal computational cost
and the presence of easy-to-apply rules for penalty selection.   In response,
this article proposes a simple new algorithm framework that requires no more
computation than a lasso path: the path of one-step estimators (POSE) does
$\ell_1$ penalized regression estimation on a grid of decreasing penalties,
but adapts coefficient-specific weights to decrease as a function of the
coefficient estimated in the previous path step.  This provides sparse
diminishing-bias regularization at no extra cost over the fastest lasso
algorithms. Moreover, our `gamma lasso' implementation of POSE is accompanied
by a reliable heuristic for the fit degrees of freedom, so that standard
information criteria can be applied in penalty selection. We also provide
novel results on the distance between weighted-$\ell_1$ and $\ell_0$ penalized
predictors; this allows us to build intuition about POSE and other
diminishing-bias regularization schemes.  The methods and results are
illustrated in extensive simulations and in application of logistic regression
to evaluating the performance of hockey players. }

\newpage
\dbl
\pagestyle{plain}
\vskip 1cm
\section{Introduction}
\label{intro}

For regression in high-dimensions, it is useful to regularize estimation
through a penalty on coefficient size.   $\ell_1$  regularization \citep[i.e.,
the lasso of][]{tibshirani_regression_1996} is especially popular, with costs
that are non-differentiable at their minima and can lead to  coefficient
solutions of exactly zero.  A related approach is concave penalized
regularization (e.g. SCAD from \citealt{fan_variable_2001} or MCP from
\citealt{zhang_nearly_2010}) with cost functions that are also spiked at zero
but flatten for large values (as opposed to the constant increase of an $\ell_1$
norm).  This yields sparse solutions where  large non-zero values are estimated
with little bias. 

The combination of  \textit{sparsity} and \textit{diminishing-bias} 
 is appealing in many settings, and a large literature on concave
penalized estimators has developed over the past 15 years.  For example, many
authors \citep{fan_variable_2001,fan_nonconcave_2004}  have contributed work on their \textit{oracle
properties}, a class of results showing conditions under which coefficient
estimates through concave penalization, or in related schemes, will be the
same as if you knew the sparse `truth' (either asymptotically or with high
probability).   From an information compression perspective,  the increased
sparsity encouraged by diminishing-bias penalties (since single large
coefficients are able to account for the signals of other correlated
covariates) leads to lower memory, storage, and communication requirements.
Such savings are especially important in distributed  computing systems
\citep[e.g.,][]{taddy_distributed_2015,gentzkow_measuring_2015}; these sorts of Big Data applications provide the original motivation behind our work in this article.

Unfortunately,  exact solvers for nonconvex penalized estimation  all require
significantly more compute time than a standard lasso.  This has precluded their use in settings -- e.g., text or web-data analysis
-- where both $n$ (the number of observations) and $p$ (covariate dimension)
   are very large. As we review in Section \ref{sec:weighted},  recent
   literature recommends the use of approximate solvers. These approximations
   take the form of iteratively-weighted-$\ell_1$ regularization, where the
   coefficient-specific weights are based upon results from previous
   iterations of the approximate solver.  Work on one-step estimation (OSE),
   e.g. by
   \cite{zou_one-step_2008}, shows that even a single step of such
   weighted-$\ell_1$ regularization is enough to get solutions that are close to
   optimal, so long as the pre-estimates are
\textit{good enough} starting points. 
The crux of success is  finding starts that are, indeed, good enough.

\begin{figure*}[tbh]
\vskip -.25cm
\includegraphics[width=\textwidth]{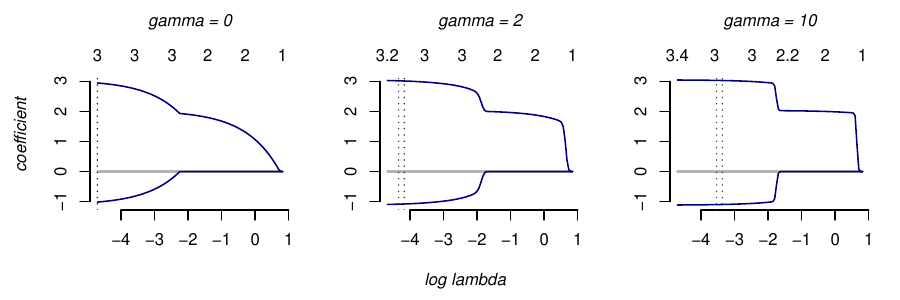}
\vskip -.25cm
\caption{\label{gamlr_eg} Gamma lasso estimation on $n=10^3$ 
 observations of $y_i = 4 + 3x_{1i} - x_{2i} + \varepsilon_i$, where
$\varepsilon_i \stackrel{ind}{\sim} \mr{N}(0,1)$ and
$\{x_{1i},x_{2i},x_{3i}\}$ are marginally standard normal with correlation of
0.9 between covariates ($x_{3i}$ is spurious). The penalty path has $T=100$
segments, $\lambda^1 = n^{-1}\left| \sum_i x_{1i}y_i\right|$, and
$\lambda^{100} = 0.01\lambda^1$. Degrees of freedom are on top and vertical
lines mark AICc and BIC selected models (see Section
\ref{sec:select}).}
\vskip -.25cm
\end{figure*}

This article provides a complete framework for sparse diminishing-bias
regularization that combines ideas from OSE with the concept of a
\textit{regularization path} -- a general technique, most famously associated
with the LARS algorithm \citep{efron_least_2004}, that estimates a sequence of
models under decreasing amounts of regularization.  So long as the estimates do not change too quickly along the path, such  algorithms can be very fast to run and are an efficient way to obtain a high-quality \textit{set} of models to choose amongst.  

A path of one-step estimators (POSE; Algorithm \ref{posealgo})
provides $\ell_1$ penalized regression  on a grid of decreasing
penalties, but adapts coefficient-specific weights to decrease as a function
of the coefficient estimated in the previous path step.  POSE takes advantage
of a natural match between path algorithms and one-step estimation: OSE
relies upon inputs being close to the optimal solution, which is precisely the
setting where path algorithms are most efficient.  We formalize `close' with a  result in Theorem \ref{thm:sparseapprox} that relates weighted-$\ell_1$ to $\ell_0$ regularization.

This framework allows us to provide 
\begin{itemize}
\item a {\it path} of coefficient fits, each element of which corresponds to sparse diminishing-bias regularization estimation under a different level of penalization; where
\item obtaining the path of coefficient fits requires {\it no more computation} than  a state-of-the-art $\ell_1$ regularization path algorithm; and
\item there are good closed-form rules for selection of the optimal penalty level along this path.
\end{itemize}
The last capability here is derived from a Bayesian interpretation for our \textit{gamma lasso} implementation of POSE, from which we are able to construct information criteria for penalty selection.
We view such tools as an essential ingredient for practical applicability in large-scale industrial machine learning where, e.g., cross-validation is not always viable or advisable.

The remainder of this paper is outlined as follows.  Section \ref{sec:algos}
presents the general regularized regression problem and introduces POSE, our
path of one-step estimators algorithm, and the gamma lasso (GL), our
implemented version of POSE.  Section \ref{sec:weighted} gives an overview on
the relationship between concave and weighted-$\ell_1$ regularization. Section \ref{sec:select} provides a Bayesian model interpretation
for the gamma lasso, and derives from this model a set of information criteria
that can be applied in penalty selection along the regularization path.
Finally, we present two empirical studies:
an extensive  simulation experiment in Section \ref{sec:sim}, and in Section \ref{sec:nhl} we investigate the data analysis question: given all goals in the past decade of NHL hockey, what can we say about individual player contributions?

\section{Paths of one-step estimators}
\label{sec:algos}

Denote $n$ response observations as $\bm{y} = [y_1,\ldots,y_n]'$ and the associated matrix of $p$ covariates as $\bm{X} =
[\bm{x}_1 \cdots \bm{x}_n]'$, with rows $\bm{x}_i = [x_{i1},\ldots,x_{ip}]'$ and columns $\bs{\chi}_j = [x_{1j},\ldots,x_{nj}]'$. Since the size of penalized $\beta_j$ depends upon the units of $x_{ij}$,  it is common to scale
the coefficient by $\mr{sd}(\bs{\chi}_j)$, the standard deviation of the $j^{th}$ column
of $\bm{X}$, before assessing its penalty cost.  We ignore this for notational convenience, but if desired you can simply replace $x_{ij}$ by $x_{ij}/\mr{sd}(\bs{\chi}_j)$
throughout. Write $\eta_{i} =
\alpha+\bm{x}_i'\bs{\beta}$ as the linear model  for observation $i$.  Denote with $l(\alpha, \bs{\beta})$, or shortened to $l(\bs{\eta})$ ,  an objective function proportional to the deviance.  For example, in Gaussian (linear)
regression, $l(\bs{\eta})$ is the sum-of-squares $0.5\sum_i \left(y_i - \eta_i\right)^2$ and in binomial (logistic)
regression,  $l(\bs{\eta}) = -\sum_i \left[\eta_iy_i -
\log(1+e^{\eta_i})\right]$ for $y_i \in [0,1]$. 

A penalized estimator is  the
solution to
\begin{equation} \label{pendev}
\argmin_{\alpha,\beta_j\in\ds{R}}~~\left\{~l(\alpha,{\bs{\beta}}) + n\lambda \sum_{j=1}^p c(\beta_j)~\right\},
\end{equation}
where $\lambda>0$ controls overall penalty magnitude and  $c(\cdot)$ is the  coefficient cost function.
 
A few common cost functions are: 
$\ell_2$, $c(\beta) \propto \beta^2$ \citep[ridge,][]{hoerl_ridge_1970}; $\ell_1$, $c(\beta) \propto |\beta|$ \citep[lasso,][]{tibshirani_regression_1996}; the `elastic net' mixture of $\ell_1$ and $\ell_2$ \citep{zou_regularization_2005}; and the log penalty $c(\beta) \propto \log(1+\gamma|\beta|)$ \citep{candes_enhancing_2008}.  Those that have
a non-differentiable spike at zero (all but ridge) lead to sparse estimators,
with some coefficients set to exactly zero.   The curvature of the penalty
away from zero dictates then the weight of shrinkage imposed on the nonzero
coefficients:  $\ell_2$ costs increase with coefficient size,  lasso's $\ell_1$
penalty has zero curvature and imposes constant shrinkage, and as curvature
goes towards $-\infty$ one approaches the $\ell_0$ penalty of subset selection.
In this article we are primarily interested in {\it concave cost functions},
like the log penalty, occupying the range between $\ell_1$ and $\ell_0$ penalties.

Penalty size, $\lambda$, acts as a {\it squelch}: it suppresses noise to
focus on the true input signal. Large $\lambda$ lead to very simple 
model estimates, while as $\lambda \rightarrow 0$ we approach maximum
likelihood estimation (MLE). Since you don't know optimal $\lambda$,
practical application of penalized estimation requires a {\it regularization
path}: a $p \times T$ field of $\bs{\hat\beta}$ estimates, say $\bs{\hat\beta}\vert_{\lambda}$, obtained while
moving from high to low penalization along $\lambda^1 > \lambda^2 \ldots >
\lambda^T$.  These paths begin at $\lambda^1$ set to infimum $\lambda$ such that
(\ref{pendev}) is minimized at $\bs{\hat\beta}\vert_{\lambda} = \bm{0}$, and move to a pre-specified $\lambda^T$ (e.g., $\lambda^T=
0.01\lambda^1$).  

Our path of one-step estimators (POSE) framework is in  Algorithm \ref{posealgo}.   

{\vskip .25cm
\begin{algorithm}[hb]
\caption{\label{posealgo} POSE }
\vskip .2cm
Initialize $\lambda^1 = \mathrm{inf}\left\{\lambda:~ \bs{\hat\beta}\vert_{\lambda} = \bf{0}\right\}$, so that $\bs{\hat\beta}_1 = \bf{0}$. ~~Say $\hat S_t = \{j: \hat\beta_j^t\neq 0 \}$

\vskip .2cm
Set step size
$0 < \delta < 1$.

\vspace{-.75cm}
\begin{align}
\text{for}~t=2\ldots T :&\notag \\
\lambda^{t} &= \delta \lambda^{t-1}\notag\\
\omega^{t}_j  &=  
\left\{ 
  \begin{array}{r}
    c'(|\hat\beta^{t-1}_j|) ~\text{for}~j \in \hat S_{t-1} \\
    1  ~\text{for}~j \in \hat S_{t-1}^c  
  \end{array} 
  \right. 
  \label{wset}\\
\left[\hat\alpha,\bs{\hat\beta}\right]^t &= \argmin_{\alpha,\beta_j\in\ds{R}}~~
l(\alpha,\bs{\beta}) + n\sum_j \lambda^t\omega^t_j|\beta_j| \label{l1pen}
\end{align}
\vskip -.3cm
\end{algorithm}\vskip .2cm}

\noindent We are assuming a cost function that is differentiable away
from zero and has been scaled such that $\lim_{b\to 0} c'(b) = 1$.  Since
POSE starts at simple $\ell_1$ penalization (i.e., with $c(\beta_j) = 1$),  our initial penalty weight is available analytically as $\lambda^1 =
n^{-1}\max\left\{\big|\partial l(\bs{\beta})/\partial
\beta_j|_{\bm{0}}\big|,~j=1\ldots p\right\}$, the maximum mean absolute
gradient evaluated at $\bs{\beta} = \bm{0}$; see the supplement for more
detail.

Section \ref{sec:weighted} details how POSE relates to concave regularization.  For some quick intuition, consider POSE with a concave cost function (e.g., the log penalty in Figure \ref{solution}).
The derivative $c'(|\hat\beta|)$ will be positive but decreasing with larger values of  $\hat\beta$, such that the \textit{weight} on the $\ell_1$ penalty for $\hat\beta_j^{t}$ will \textit{diminish} with the size of $|\hat\beta_j^{t}|$.  This implies that coefficient estimates later in the path will be less biased towards zero if that coefficient has a large value earlier in the path.

\begin{figure*}[t]
\vskip -.25cm
\includegraphics[width=\textwidth]{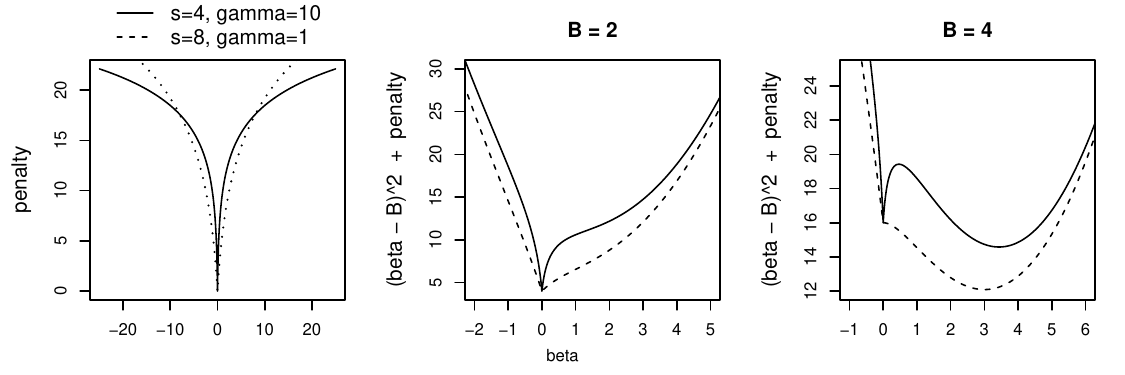}
\vskip -.25cm
\caption{\label{solution} Log penalties $c(\beta) = s\log(1 + \gamma|\beta|)$ 
and penalized objectives $(\beta-B)^2 + c(\beta)$.}
\end{figure*}

\subsection{The gamma lasso}
\label{sec:gamlr}

The gamma lasso (GL) specification for POSE is based upon the log penalty,
\begin{equation}\label{logpen}
c(\beta_j) =  \gamma^{-1}\log(1+\gamma|\beta_j|),
\end{equation} where $\gamma > 0$. 
This appears under a variety of parameterizations and names in the literature;
see \citet{mazumder_sparsenet_2011} and applications in
\citet{friedman_fast_2008}, \citet{candes_enhancing_2008},
\citet{cevher_learning_2009}, \citet{taddy_multinomial_2013} and
\citet{armagan_generalized_2013}.   The penalty is nondifferentiable at zero
and concave away from zero with curvature $c''(|b|) =
-\gamma/(1+\gamma|b|)^2$ and gradient $c'(|b|) = 1/(1+\gamma|b|)$ (note that
$\lim_{b\to 0} c'(b) = 1$ as required).  This  spans
the range from $\ell_1$ to $\ell_0$  penalization: 
$\lim_{\gamma\to 0} c(b) = |b|$, 
while large $\gamma$  yield $\ell_0$-type costs with $c'(|b|) \approx 0~\forall b\neq 0$. 
(Note that $\lim_{\gamma\to\infty} c(b) = 0$; however, POSE yields forward stepwise selection in this limit).

GL simply replaces line (\ref{wset}) in Algorithm \ref{posealgo}  with
\begin{equation}\label{gammalasso} \omega^{t}_j  = \left(1 + \gamma
|\hat\beta^{t-1}_j|\right)^{-1} ~~j=1\ldots p  \end{equation} Behavior of the
resulting paths  is governed by $\gamma$, which we refer to as the penalty
{\it scale}.  Under $\gamma=0$, GL is just the usual lasso.  Bias diminishes
faster for larger $\gamma$.  At the  extreme, $\gamma=\infty$ yields a greedy
subset selection routine where a coefficient is unpenalized in all segments
after it first becomes nonzero.   Figure \ref{gamlr_eg} shows solutions in a
simple problem.

Each  gamma lasso path segment is solved through coordinate descent (see supplement).  The algorithm is implemented in {\tt
c} as part of the {\tt gamlr} package for {\sf R}. 
Usage of {\tt gamlr} mirrors that of its
convex penalty analogue {\tt glmnet} \citep{friedman_regularization_2010}, the
fantastic and widely used package for costs between $\ell_1$ and $\ell_2$ norms. In
the lasso case ($\gamma=0$), the two algorithms are essentially equivalent.

\section{Weighted-$\ell_1$ approximations to concave penalization}
\label{sec:weighted}

Concave penalties such as the log penalty, which have a gradient that is
decreasing with absolute coefficient size,  yield the `diminishing-bias'
property discussed above. It is {\it the} reason why one would
use concave penalization instead of $\ell_1$ or convex alternatives.

Unfortunately, such penalties can overwhelm the convex likelihood and produce a
nonconvex minimization objective; see Figure \ref{solution}.
This  makes computation difficult.  For example,  one run of SCAD via the
\texttt{ncvreg} R package \citep{breheny_coordinate_2011} for the simulation in
Section \ref{sec:sim} requires around 10 minutes, compared to less than a second for 
lasso (or gamma lasso).  The most efficient exact solver that we've found is the
\texttt{sparsenet} of \citet{mazumder_sparsenet_2011}, also implemented in R,
which first fits a lasso path and, for
each segment on this path, adapts coefficient estimates along a second path of
increasing penalty concavity. However, \texttt{sparsenet}  relies upon sequential solution
over a large set of specifications 
 and its compute cost remains much higher than for the [gamma] lasso.

Local linear approximation \cite[LLA; e.g.,][]{candes_enhancing_2008}
algorithms replace the nonconvex cost function $c$  with its tangent at the
current estimate, $c'(\hat\beta_j)\beta_j$.  The objective is then just a
weighted $\ell_1$ penalized loss. An exact LLA solver iterates between updating
$c'(\hat\beta)$ and solving the implied $\ell_1$ penalized minimization problem.
\citet{zou_one-step_2008} present numerical and theoretical evidence that LLA
can provide near-optimal solutions even if you {\it stop it
after one iteration}. This is an example of one-step estimation \cite[OSE;][]{bickel_one-step_1975}, wherein you take as
your estimator the first step of an iterative approximation to some objective.
Early-stopping can be as good  as the full solution {\it if} the
initial estimates are good enough.

OSE and similar ideas have had a resurgence in the concave penalization
literature recently, motivated by the need for faster estimation algorithms.
\cite{fan_strong_2014} consider early-stopping of LLA for folded concave
penalization  and show that, under strong sparsity assumptions about true
$\bs{\beta}$ and given appropriate initial values, OSE LLA is with high
probability an oracle estimator.   Zhang (2010,2013)
\nocite{zhang_analysis_2010,zhang_multi-stage_2013} investigates  
`convex relaxation' iterations, where estimates under convex regularization
 are the basis for weights in a subsequent penalized objective.
 \cite{wang_calibrating_2013} propose a two step algorithm that feeds lasso
 coefficients into a linear approximation to folded concave penalization.
 These OSE methods are all closely related to  the adaptive lasso
\citep[AL;][]{zou_adaptive_2006}, which does weighted-$\ell_1$ minimization under
weights $\omega_j = 1/|\hat\beta^0_j|$, where $\hat\beta^0_j$ is an initial
guess at the coefficient value.  The original AL paper advocates using MLE
estimates for initial values, while
\cite{huang_adaptive_2008} suggest using marginal correlations
$\hat\beta^0_j = \mr{cor}(\bm{x}_j,\bm{y})$; this marginal AL algorithm is included in our simulations of Section \ref{sec:sim}.

The main point is that OSE LLA, or a two-step estimator starting from
$\bs{\hat\beta}=0$, or any version of the adaptive lasso, are all  interpretable as weighted-$\ell_1$ 
penalization with weights equal to something like $c'(\beta^0)$ for initial
coefficient guess $\beta^0$. The algorithms proposed in Section \ref{sec:algos}, POSE and GL, take advantage of an  available
source of initial values in any path estimation algorithm -- the solved values
from the previous path iteration.  Our  simulations in Section \ref{sec:sim}
show that this efficient strategy works as well or better than  expensive
exact solvers.  In the next section, we provide some theoretical intuition on
why it works.

\subsection{Comparison between weighted-$\bs{\ell_1}$ and $\bs{\ell_0}$ penalization}

Our oracle comparator is estimation under $\ell_0$ costs, $c(\beta_j) =
\ds{1}_{\{\beta_j\neq0\}}$, for which global solution is impractical.  We
treat the design $\bm{X}$ as {\it fixed}, and make no assumptions
about the distribution for $\bm{y} | \bm{X}$ (not even independence between
observations; see remarks).   The question we address is thus more
operational than statistical: for a fixed sample, what is the distance between
an easy-to-find weighted-$\ell_1$ penalized solution and the infeasible
$\ell_0$-penalized optimum?

In the theoretical setups more familiar to statisticians,
one bounds the expected difference  either between fitted coefficients
and some assumed `true' model parameters or  between model predictions and future
response.  Both of these evaluations are important, and they have been studied
extensively elsewhere. The text by \cite{buhlmann_statistics_2011}
includes a comprehensive treatment of such prediction and estimation risk for
$\ell_1$ and weighted-$\ell_1$ penalized linear models,  and we refer the
interested reader there for this material and abundant references to other
relevant work.  This literature usually assumes a truth that is (at least
approximately) linear and sparse in the available covariates.  The assumption
of true sparsity is dubious in many realistic applications, but it is
necessary inasmuch as, with finite data, most parameters in a
high-dimensional model cannot be reliably measured as different from zero.

Instead, we are mostly interested in weighted-$\ell_1$ penalization as a way
to obtain fits that are as sparse as possible without compromising prediction, regardless of whether the data generating process is sparse.  By
comparing to an ideal optimization objective rather than to some true model,
we are able to present a finite-sample fully-nonparametric result with a
straightforward intuitive proof. Obviously, this result is useful only if
an $\ell_0$ penalized estimator would work well for the
problem at hand. However,  there is theoretical support for optimality of
$\ell_0$ penalization in a  variety of high-dimensional prediction settings
\citep[e.g,][]{mallows_comments_1973,efron_estimation_2004}.  Indeed, our
simulation studies show that $\ell_0$ oracles are nearly always the best
performing option in terms of both prediction and estimation error.  In the
case where an $\ell_0$ oracle does poorly, we would usually argue against
penalized linear models as a strategy anyways.

\subsubsection{Approximation to an $\bs{\ell_0}$ oracle}

 For  $S \subset \{1\ldots p\}$ with cardinality $|S|=s$ and complement $S^c =
\{1\ldots p\}\setminus S$, denote vectors restricted to covariates in $S$ as
$\bm{\beta}_S = [\beta_j:j\in S]'$, matrices as $\bm{X}_S$, etc.  Use
$\bs{\beta}^S$ to denote the coefficients for ordinary least-squares (OLS)
restricted to $S$: that is, $\bs{\beta}^S_S =
(\bm{X}_S'\bm{X}_S)^{-1}\bm{X}_S'\bm{y}$ and $\beta^{S}_j = 0~\forall~j\notin
S$.  Moreover, $\bm{e}^S = \bm{y}-\bm{X}\bs{\beta}^S =
(\bm{I}-\bm{H}^S)\bm{y}$ are residuals and $\bm{H}^S =
\bm{X}_S(\bm{X}_S'\bm{X}_S)^{-1}\bm{X}_S'$ the projection matrix from OLS on $S$.  
Use $|\cdot|$ and $\|\cdot\|$ for $\ell_1$ and $\ell_2$ norms.

We use the following result for iterative \textit{stagewise} regression; proof is in the supplement. 
\begin{lemma}\label{SSElemma}
Say $\mr{MSE}_S = \|\bm{X}\bs{\beta}^S-\bm{y}\|^2/n$ and 
$\mr{cov}(\bs{\chi}_j,\bm{e}^S) = \bs{\chi}_j'(\bm{y}-\bm{X}\bs{\beta}^S)/n$ are sample variance and covariances.  Then for any $j \in 1\ldots p$, 
\[
\mr{cov}^2(\bs{\chi}_j,\bm{e}^S) \leq \mr{MSE}_S - \mr{MSE}_{S\cup j}
\]
\end{lemma}

In addition, we need to define {\it restricted eigenvalues} (RE) on the gram matrix $\bm{X}'\bm{X}/n$.
  This  RE  matches the `adaptive restricted eigenvalues' of \cite{buhlmann_statistics_2011}.  
Similar quantities are common in the theory of regularized estimators; \cite{raskutti_restricted_2010} show that similar conditions hold given  $\omega^{\mr{min}}_{S_c}=1$ with high probability for $\bm{X}$ drawn from a broad class of Gaussian distributions.  \cite{bickel_simultaneous_2009} provide a nice overview of sufficient conditions, and \cite{buhlmann_statistics_2011} have extensive discussion and examples.  
\begin{definition}\label{redef}
The restricted eigenvalue is
$
\phi^2(L,S) = \min_{\{\bm{v}: \bm{v}\neq \bm{0},~|\bm{v}_{S^c}| \leq L\sqrt{s}\|\bm{v}_S\|\}}\frac{\|\bm{X}\bm{v}\|^2}{n\|\bm{v}\|^2}$.
\end{definition}

\noindent Finally,  we bound the distance between prediction
from $\ell_0$ and weighted-$\ell_1$ regularization.  

\begin{theorem} \label{thm:sparseapprox}  Consider squared-error loss
$l(\bs{\beta}) =
\frac{1}{2}\|\bm{X}\bs{\beta}-\bm{y}\|^2$, and suppose $\bs{\beta}^{S}$ minimizes the $\ell_0$ penalized objective $l(\bs{\beta}) + n\nu\sum_{j=1}^p\ds{1}_{\{\beta_j\neq0\}}$ with  $|S|=s<n$.   
Write $\bs{\hat\beta}$ as solution to the weighted-$\ell_1$ minimization $l(\bs{\beta}) + n\lambda\sum_j\omega_j|\beta_j|$. 

Then  
$\omega^{\mr{min}}_{S^c}\lambda > \sqrt{2\nu}$ while $\phi^2(L,S) > 0$, with 
 $L = \frac{\|\bs{\omega}_S\|}{\sqrt{s}}\left(\omega^{\mr{min}}_{S^c}-\sqrt{2\nu}/\lambda\right)^{-1}$, implies
\begin{equation} \label{sparseineq}
\frac{\|\bm{X}(\bs{\hat\beta}-\bs{\beta}^S)\|^2}{n}\leq
\frac{4\lambda^2 \|\bs{\omega}_S\|^2}{\phi^2(L, S)}.
\end{equation} 
\end{theorem}

\begin{proof}
From the definitions of $\bs{\hat\beta}$ and $\bs{\beta}^S$, 
\begin{align}\label{predconv}
\frac{1}{2}\|\bm{X}\bs{\hat \beta}-\bm{y}\|^2 +  n\lambda\sum_j\omega_j|\hat\beta_j|   &= \frac{\|\bm{X}(\bs{\hat\beta}-\bs{\beta}^S)\|^2}{2} + \frac{\|\bm{e}^S\|^2}{2} - \bm{\hat y}'\bm{e}^S+ n\lambda\sum_j\omega_j|\hat\beta_j| 
\\&\leq~~\frac{1}{2}\|\bm{e}^S\|^2 + n\lambda\sum_{j\in S}\omega_j|\beta^S_j|\notag
\end{align}
  Since $\bm{\hat y}'\bm{e}^S = \bm{\hat y}'(\bm{I}-\bm{H}^S)\bm{y} =
\bs{\hat\beta}'\bm{X}'(\bm{y}-\bm{X}\bs{\beta}^S) = 
\sum_{j\in S^c} \hat\beta_j\bs{\chi}_j'(\bm{y}-\bm{X}\bs{\beta}^S)
$,
 we can apply Lemma \ref{SSElemma} followed by $\bs{\beta}^S$ being optimal under $\ell_0$ penalty $\nu$ to get 
\begin{equation} \label{L0ineq}
\left(\frac{\bs{\chi}_j'(\bm{y}-\bm{X}\bs{\beta}^S)}{n}\right)^2
\leq \mr{MSE}_S - \mr{MSE}_{S\cup j} < 2\nu ~~~\forall~j
\end{equation}
so that $|\bm{\hat y}'\bm{e}^S| = |\bs{\hat\beta}_{S^c}\bm{X}_{S^c}'(\bm{y}-\bm{X}\bs{\beta}^S)| < n\sqrt{2\nu}|\bs{\hat\beta}_{S^c}|$.  Applying this inside (\ref{predconv}),
\begin{align}\label{predconvineq}
\frac{1}{2}\|\bm{X}(\bs{\hat\beta}-\bs{\beta}^S)\|^2
  + n\left(\omega^{\mr{min}}_{S^c}\lambda-\sqrt{2\nu}\right)|\bs{\hat\beta}_{S^c}|
  ~\leq~ n\lambda\sum_{j\in S}\omega_j|\hat\beta_{j}-\beta^S_j|
  ~\leq~ n\lambda\|\bs{\omega}_S\|\|\bs{\hat\beta}_{S}-\bs{\beta}^S_S\|.
\end{align}
Given $\omega^{\mr{min}}_{S^c}\lambda > \sqrt{2\nu}$,
difference $\bs{\hat\beta}-\bs{\beta}^S$ is in the RE support for 
$L=\frac{\|\bs{\omega}_S\|}{\sqrt{s}}(\omega^{\mr{min}}_{S^c}-\sqrt{2\nu}/\lambda)^{-1}$ and thus $\|\bs{\hat\beta}_{S}-\bs{\beta}^S_S\| \leq \|\bm{X}(\bs{\hat\beta}-\bs{\beta}^S)/\sqrt{n}\|/\phi(L,S)$.  Finally, applying this inside (\ref{predconvineq}) yields
\begin{align}
\frac{1}{2}\|\bm{X}(\bs{\hat\beta}-\bs{\beta}^S)\|^2
  &~\leq~ \frac{\sqrt{n}\lambda\|\bs{\omega}_S\|\|\bm{X}(\bs{\hat\beta}-\bs{\beta}^S)\|}
  {\phi(L, S)}.
\end{align}
Dividing each side by $\sqrt{n}\|\bm{X}(\bs{\hat\beta}-\bs{\beta}^S)\|/2$ and squaring gives the result.
\end{proof}

\noindent {\bf Remarks}

\noindent \vskip .2cm \noindent $\bullet$~~  Theorem \ref{thm:sparseapprox} is finite sample exact.  Distinguishing it from
related results in the literature, it is  also completely
\textit{non-parametric} -- it makes no reference to the true distribution of
$\bm{y}|\bm{X}$.  Indeed, if we make such assumptions, Theorem
\ref{thm:sparseapprox} provides bounds on the distance between a
weighted-lasso and optimal prediction.  The next remark is an example.

\noindent \vskip .2cm \noindent $\bullet$~~ Assume that $\bm{y} \sim
(\bs{\eta},\sigma^2\bm{I})$ --  $y_i$ independent with mean $\mu_i$ and shared
variance $\sigma^2$.  The $C_p$ formula of
\cite{mallows_comments_1973} gives 
$\mr{MSE}_S + 2s\sigma^2/n$ as an unbiased estimate of residual variance.
Following \cite{efron_estimation_2004}, this implies $\nu = \sigma^2/n$ is the
optimal $\ell_0$ penalty for minimizing prediction error. Theorem
(\ref{thm:sparseapprox}) applies directly,  with $L =
\frac{\|\bs{\omega}_S\|}{\sqrt{s}}\left(\omega^{\mr{min}}_{S^c}-
\sqrt{2}\sigma/(\lambda\sqrt{n})\right)^{-1}$, to give a bound on the distance
between weighted-$\ell_1$ estimation and $C_p$-optimal prediction. Note that,
since the condition on minimum $S^c$ weights has become
$\omega^{\mr{min}}_{S^c} > (\sigma/\lambda)\sqrt{2/n}$, comparison to $C_p$
suggests we can use larger $\gamma$  with large $n$ or small $\sigma$.

\noindent \vskip .2cm \noindent $\bullet$~~  Plausibility of the restricted eigenvalue assumption
$\phi(L,S) > 0$ depends  upon $L$.  It is less restrictive if we can reduce
$\|\bs{\omega}_S\|$ without making $\omega^{\mr{min}}_{S_c}$ small.
\textit{This is a key motivation for the POSE algorithm:} if covariates with
nonzero $\hat \beta_j$ for large $\lambda$ (i.e., early in the path) can be
assumed to live in $S$, then increasing $\gamma >0$ will improve prediction.
Of course, the moment that $\lambda$ becomes small enough that elements of
$S_c$ get nonzero $\hat\beta$, then larger $\gamma$ can lead to over-fit.  For
this reason it is essential that we have tools for choosing  optimal
$\lambda$.   In the following sections, we describe both cross-validation and information criteria
for penalty selection.

\noindent \vskip .2cm \noindent $\bullet$~~  For the lasso, $\bs{\omega}=\bm{1}$ and $\|\bm{X}(\bs{\hat\beta}-\bs{\beta}^S)\|^2/n \leq
4\lambda^2 s/\phi^2(L, S)$ with $L = (1 - \sqrt{2\nu})^{-1}$.  This bound depends only upon $s$, but forcing $\phi^2(L,S) > 0$ becomes more restrictive for larger $p$.

\noindent \vskip .2cm \noindent $\bullet$~~ There is no notion of a `true'
model here and this result has nothing to say about the estimation error on
individual parameters.  We again refer the reader to
\cite{buhlmann_statistics_2011}, and note that for minimizing 
estimation error they recommend $\ell_1$ weights derived
from lasso fits at smaller $\lambda$ (i.e., the opposite of GL and POSE, where weights are derived from fits at slightly larger
$\lambda$).   However, we  consider estimation error in our simulation study
and find that the GL algorithms do well compared to MCP and the adaptive
lasso. Moreover, in the supplemental material we adapt standard results from
\cite{wainwright_sharp_2006,wainwright_sharp_2009} to show how reducing
$\|\bs{\omega}_S\|$ without making $\omega^{\mr{min}}_{S_c}$ small can lead to
lower false discovery rates with respect to the $\ell_0$ oracle.
Unfortunately, this relies upon fairly strict design restrictions.

\section{Penalty selection}
\label{sec:select}

Lasso, the gamma lasso, and related 
sparse regularization estimators do not actually {\it do} model selection; rather, they
can be used to obtain paths of estimates corresponding to different levels of
penalization.  Each penalty level corresponds to a different `model'
and we must select the optimal choice from these candidates.

$K$-fold cross-validation \cite[CV; e.g.,][]{efron_estimation_2004}  is the
most common technique for penalty selection, and it does a good job.  However,
there are many scenarios where we might want an analytic alternative to CV.
For example, if a single fit is expensive then doing it $K$ times will be
impractical.  More subtly, truly Big Data are distributed.  Algorithms can be
designed to work in parallel on subsets
\citep[e.g.,][]{taddy_distributed_2015} but a bottleneck results if you need
to communicate across machines for CV experimentation. Finally, CV can lead to
over-fit for unstable algorithms whose results change dramatically in response
to data jitter; see \cite{breiman_heuristics_1996} for a classic discussion and the supplement for an overview.

An important feature of the standard $\ell_1$ lasso is that it comes with a
simple approximation for the estimation degrees-of-freedom (\textit{df}) at
any $\lambda$: the number of nonzero estimated coefficients
\citep[see][]{zou_degrees_2007}.  These \textit{df} can combined with the fitted
deviance  to create information criteria, such as the AIC or BIC, that provide alternative tools for penalty selection.

This section derives the gamma lasso as approximately maximizing the posterior
for a hierarchical Bayesian model, and uses this
interpretation to obtain a heuristic degrees-of-freedom for each estimate
along the GL path.  These GL \textit{df} can  be input to information criteria
for model selection. In particular, our extensive simulations demonstrate that the GL
\textit{df} can be used with  the corrected AICc of
\citet{hurvich_regression_1989}  to obtain out-of-sample predictive
performance that is as good or better than that from cross-validated
predictors.

\subsection{Bayesian model interpretation}

Consider a model where each $\beta_j$ is
assigned a Laplace distribution prior with scale $\tau_j>0$,
\begin{equation}\label{blasso}
\beta_j \sim \mr{La}\left(\tau_j\right) =
\frac{\tau_j}{2}\exp\left[ -\tau_j|\beta_j| ~\right].
\end{equation}
Typically, scale parameters $\tau_1 =
\ldots = \tau_p$ are set as a single value, say $n\lambda/\phi$ where
 $\phi$ is the dispersion (e.g. Gaussian variance
$\sigma^2$ or 1 for the binomial).   Posterior
maximization under the prior in (\ref{blasso}) is $\ell_1$ regularized estimation \citep[e.g.,][]{park_bayesian_2008}.

Instead of a single shared scale, assume an independent gamma
$\mr{Ga}(s,1/\gamma)$ hyperprior with `shape' $s$ and `scale' $\gamma$ for
each $\tau_j$, such that $\ds{E}[\tau_j] = s\gamma$ and $\mr{var}(\tau_j) =
s\gamma^2$.  The {\it joint} coefficient-scale prior is
\begin{align}\label{glprior}
\pi(\beta_j,\tau_j) = \mr{La}\left(\beta_j ;~ \tau_j\right)
\mr{Ga}\left(\tau_j;~ s,\gamma^{-1}\right) = \frac{ 1}{2\Gamma({s})} 
\left(\frac{\tau_j}{\gamma}\right)^{s}
               \exp\left[-\tau_j(\gamma^{-1}+|\beta_j|)\right].
\end{align}
The gamma hyperprior is conjugate here, implying a $\mr{Ga}\left(s+1, ~1/\gamma +
|\beta_j|\right)$ posterior for $\tau_j \mid \beta_j$ with conditional
posterior mode (MAP) at $\hat\tau_j = \gamma s/(1 + \gamma |\beta_j|)$.

Consider joint MAP estimation of $[\bs{\tau},\bs{\beta}]$ under the prior in
   (\ref{glprior}), where we've suppressed $\alpha$ for simplicity. By taking
   negative logs and removing constants, this is equivalent to solving
\begin{equation}\label{gljoint}
\argmin_{\beta_j\in\ds{R},~\tau_j \in \ds{R}^{+}}\!\!
\frac{l(\bs{\beta})}{\phi} + \sum_j \left[\tau_j(\gamma^{-1}+|\beta_j|) - s\log(\tau_j)\right].
\end{equation}
It is straightforward to show (supplement) that the $\bs{\beta}$ which solves (\ref{gljoint})  is also the solutions 
to the log-penalized objective 
\begin{equation}\label{logobj}
\argmin_{\beta_j\in\ds{R}}~~
\phi^{-1}l(\bs{\beta}) + \sum_j  s\log(1+\gamma|\beta_j|),
\end{equation}
such that the log penalty is interpretable as a {\it profile} MAP estimate.

\subsection{Degrees of freedom}

For prediction rules, say $\hat y_i$, that are suitably stable \citep[i.e.,
Lipschitz; see][]{zou_degrees_2007}, the SURE framework of
\cite{stein_estimation_1981} applies and  $df =
\ds{E}\left[\sum_i \partial \hat y_i/\partial y_i\right]$.
Consider  a single coefficient $\beta$ estimated via least-squares under $\ell_1$
penalty $\tau$.   Write gradient at zero $g = -\sum_i x_iy_i$ and curvature $h
= \sum_i x_i^2$ and set $\varsigma = -{\tt sign}(g)$. The prediction rule is
$\hat y = x(\varsigma/h)(|g|-\tau)_+$ with  derivative $\partial\hat y_i/\partial y = x_i^2/h \ds{1}_{[|g|>\tau]}$, so that the SURE expression
yields $df = \ds{E}\left[ \ds{1}_{[|g|>\tau]} \right]$.   This expectation is
taken with respect to the {\it unknown true} distribution over $\bm{y} |
\bm{X}$, not that estimated from the observed sample.  However, 
 one can evaluate this expression at observed
gradients as an unbiased estimator for the true \textit{df} \citep[e.g.,][]{zou_degrees_2007}.

This motivates our heuristic $df$ in weighted-$\ell_1$ regularization:  the {\it prior} expectation for the number  $\ell_1$ penalty dimensions, $\tau_j = \lambda \omega_j$, that are less than their corresponding absolute gradient dimension.  Referring to our Bayesian model above, each $\tau_j$ is $iid$ $\mr{Ga}(s,1/\gamma)$ in the prior,
leading to the GL degrees of freedom
\begin{equation}
\label{edf} df^t = \sum_j \mr{Ga}(|g_{j}|;~n\lambda^t/(\gamma\phi),
1/\gamma), \end{equation} where $\mr{Ga}(~\cdot~;~\mr{shape}, 1/\mr{scale})$
is the Gamma cumulative distribution function and $g_j$ is an estimate of  the
$j^{th}$ coefficient gradient evaluated at $\hat\beta_j=0$. Note that
the number of unpenalized variables (e.g., 1 for $\alpha$) should also be
added to the total  estimation $df$.  For  orthogonal covariates,  $g_j$ is
just the marginal gradient at zero. In the non-orthogonal case, where $g_{j} =
g_j(0)$ becomes a function of all of the elements of $\bs{\hat\beta}$, we plug
in the most recent $g_j$ at which $\hat\beta^t_j=0$:  this requires no extra
computation and has the advantage of maintaining $df^t =
\hat p^t$ for $\gamma = 0$.

\subsection{Selection via information criteria}

An information criterion is an attempt to approximate
 divergence between the unknown true data generating
process and our parametric approximation; see the supplement for an
overview.  These  take the form
\begin{equation}
l(\bs{\hat\beta}) + k(df)
\end{equation}
where $k$ is the cost on the degrees-of-freedom associated with
 $\bs{\hat\beta}$ and $l$ is the negative log likelihood.  The AIC of
\cite{akaike_information_1973} uses $k(df) = df$ while the BIC of \cite{schwarz_estimating_1978}
uses $k(df) = \log(n)df/2$. As detailed in \cite{flynn_efficiency_2013}, the
corrected AICc with $k(df) = df\times n/(n-df-1)$ does a better job than the
AIC or BIC in choosing the optimal model for prediction when $df$ is large.
Alternatively, the BIC is often preferred for accurate support recovery or
avoiding false discovery; see, e.g.,
\cite{zou_degrees_2007}.

\section{Simulation experiment}
\label{sec:sim}

We consider continuous-response data simulated from
 a $p=1000$ dimensional linear model
\begin{align}\vspace{-.1cm}
\label{simdgp}
y &\sim \mr{N}\left(\bm{x}'\bs{\beta},\sigma^2\right) ~~\text{where}~~
\beta_j = (-1)^j\exp\left(-\frac{j}{{\kappa}} \right)\ds{1}_{[j\leq J]}~~\text{for}~~j=1\dots p~~.
\end{align}\vspace{-.1cm}
Each simulation draws $n$ means $\eta_i =
\bm{x}_i'\bs{\beta}$ and two independent samples 
$\bm{y},\bm{\tilde y} \sim \mr{N}(\bs{\eta},\sigma^2\bm{I})$;  the first
sample is used for training and we evaluate prediction error on the second
sample.  Our experiment includes all possible combinations of the following configuration options:
\begin{itemize}\vspace{-.1cm}\setlength\itemsep{0em}
\item the sample size is $n=100$ or $n=1000$;
\item the simulation models is either {\it dense}, with $J=p$ so that all true coefficients are nonzero, or 
  {\it sparse}, with $J = n/10$ for either 10 or 100 nonzero coefficients;
\item defining $\bm{z}_i \sim \mr{N}\left(\bm{0},\bs{\Sigma}\right)$ for $i=1\ldots n$, the regression inputs $\bm{x}_i$ are generated as either \textit{continuous} $x_{ij}=z_{ij}$  or \textit{binary}  $x_{ij} \stackrel{ind}{\sim} \mr{Bern}\left( 1/(1+e^{-z_{ij}})\right)$;
\item error variance $\sigma^2$ is defined through {\it signal-to-noise} (\textit{s2n})
ratios $\mr{sd}(\bs{\eta})/\sigma$ of $1/2$, $1$, or $2$;
\item design multicollinearity is parametrized via $\Sigma_{jk} =
\rho^{|j-k|}$, with $\rho$ of $0$, $0.5$, or $0.9$;
\item the rate of coefficient decay is specified by
${\kappa}$ of $10$, $50$, $100$, or $200$.  
\end{itemize}\vspace{-.1cm}
This implies a total of 288 different models, and we simulate and estimate 1000 times for each.

\begin{figure}[tb]
{\small \hskip 2.5cm$\bs{\gamma=0}$\hskip 4cm$\bs{\gamma=1}$\hskip 4cm$\bs{\gamma=10}$}

\includegraphics[width=\textwidth]{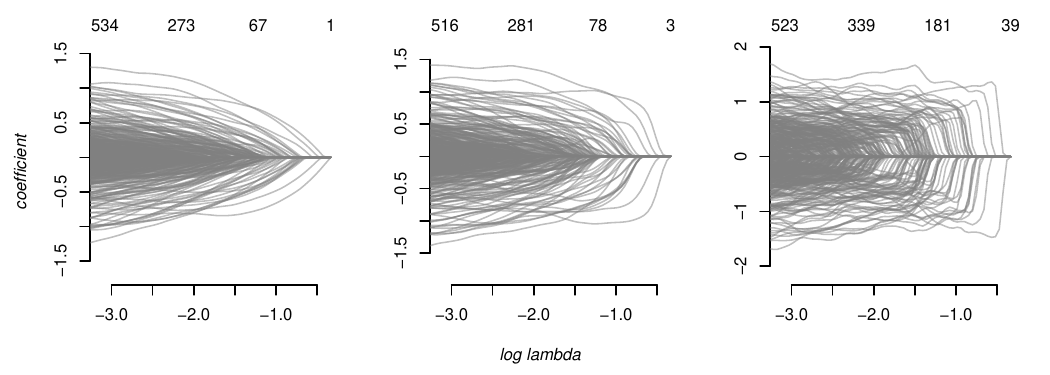}
\caption{\label{simpaths} Regularization paths for simulation example.
Degrees of freedom $df^t$ are along the top. }
\vskip .3cm
{\small \hskip 2.5cm$\bs{\gamma=0}$\hskip 4cm$\bs{\gamma=1}$\hskip 4cm$\bs{\gamma=10}$}

\includegraphics[width=\textwidth]{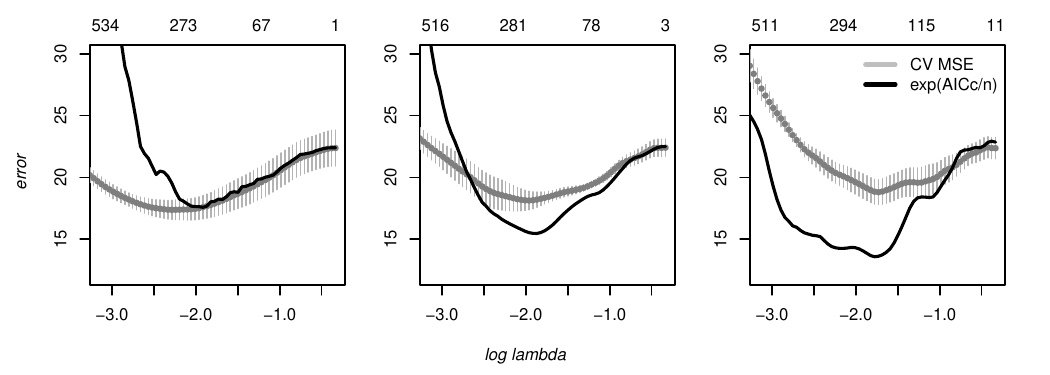}
\caption{\label{simcv} 5-fold CV and AICc for a simulation example. 
Points-and-bars show mean OOS MSE $\pm 1$se.  }
\end{figure}

Our simulation models include various levels for both true sparsity --
adjusted through the threshold $J$ -- and the {\it effective sparsity} dictated by
$\kappa$, which controls the rate of coefficient decay. At ${\kappa}=10$ only
23 coefficients have absolute value larger than 0.1, while at ${\kappa=200}$
there are 460 coefficients larger than this threshold. The strictly dense $J=p$
models have all nonzero  true coefficients but it will be useless to estimate many of them
when $p=n$ or $p>n$.

Figures \ref{simpaths} and \ref{simcv} illustrate GL paths for a single
dataset, generated from our {\it dense} model with {\it binary} design,   $\mr{sd}(\bm{\eta})/\sigma=1$,  $\rho=0.9$, and  ${\kappa}=50$. In Figure
\ref{simpaths}, increasing $\gamma$ leads to `larger shouldered' paths where
estimates move quickly to MLE for the nonzero-coefficients. Degrees of
freedom, calculated as in (\ref{edf}), are along the top of each plot; at all but the smallest $\lambda$ values, equal
$\lambda$ have higher $df^t$ for higher $\gamma$ since there is less shrinkage
of $\hat\beta_j\neq0$.   This relationship switches only when $df^t$
nears $n$, indicating that the heuristic in (\ref{edf}) might underestimate $df$ for clearly over-fit models.
Figure \ref{simcv} shows CV and AICc error estimates and we see that 
the two criteria roughly track each other.  Notice that
for larger $\gamma$ the CV error increases more quickly away from
its minimum; this is predicted by Theorem \ref{thm:sparseapprox} and shows that the consequences of over fit are worse with
faster diminishing-bias.  Computation times also increase with larger $\gamma$, although a
single run at $\gamma=10$ still requires less than a second.

Our simulation experiment includes \texttt{gamlr} runs of GL with $\gamma$
of 0 (lasso), 1, and 10. For penalty selection, we focus on AICc with the GL $df$
from Section \ref{sec:select} as well as 5-fold CV (with $\lambda$ selected to
{\it minimize} CV error estimates); results for additional selection rules are in the supplement. 
We also consider `GL-select', a routine which chooses $\gamma \in \{0,1,10\}$ to minimize these AICc or CV values.
Predictive performance is measured through the average root
mean square error (RMSE) for  $\bm{\hat y}$ on $\bm{\tilde y}$, and RMSE values are reported  in
terms of percentage-worse than the oracle-support MLE procedure.  We include the
oracle regression $R^2$ for reference.

We include performance results for MLE fit on `oracle' restricted support.  For the strictly sparse model, with $J=n/10$, our oracle uses the true nonzero support.  For the strictly dense model, our oracle comparator is the $C_p$ optimal $\ell_0$ penalized solution
\begin{equation}\label{l0oracle}
\bs{\beta}^{\star} = \argmin_{\bs{\beta}} \left\{ \|\bm{y}-\bm{X}\bs{\beta}\|^2 + 2\sigma^2\sum_j
\ds{1}_{\{\beta_j\neq0\}}\right\},
\end{equation} which we solve by searching through
OLS on $\bm{X}_{\{1\ldots j\}}$ for $j=1\ldots p$ (since the true coefficients are ordered).
See \cite{mallows_comments_1973} and remarks after Theorem \ref{thm:sparseapprox} for background on this oracle.

As a `cheap' comparator, with run-time similar to that of GL, we
consider the {\it marginal} adaptive lasso with $\ell_1$ penalty
weights $\omega_j
\propto |\mr{cor}(\bm{x}_j,\bm{y})|^{-1}$.   The AL weights are scaled so that
$\mathrm{min}(\omega_j)=1$ and we set $df$ as
the number of nonzero estimated parameters.   As an `expensive' gold standard,
we include an exact solver for MCP penalized regression. Described in Section \ref{sec:weighted}, {\tt sparsenet} applies 5-fold CV to
optimize out-of-sample error over a dense grid of potential penalty sizes
($\lambda$) and concavities (analogous to our $\gamma$) for the MCP penalty.

Table \ref{tab:sumtables} presents a summary of predictive performance over
all simulation models, Figure \ref{fig:weights}
plots some of the quantities from Theorem \ref{thm:sparseapprox} for different algorithms, and Table \ref{tab:esttables} summarizes estimation error against true coefficients in our sparse simulation model.
These results  represent only a small portion of the simulation
study.  We have aggregated across different covariate designs (binary or
continuous, with various levels of multicollinearity) and have combined
results from $\kappa \in \{10,50\}$ as  `fast' decay and those for $\kappa \in
\{100,200\}$ as `slow' decay.  
The supplement contains an additional
128 tables, detailing hundreds of data generating processes and
algorithm configurations, with results on prediction and estimation error, on the fitted number of nonzero
parameters, and on sensitivity and false discovery rates.

\begin{table}
\footnotesize
\begin{center}
\begin{tabular}{cc|cc|cc|cc|cc|cc|c|c}
& & \multicolumn{11}{l|}{\bf \% Worse than Oracle RMSE } & \\[1ex]
& \multirow{2}{*}{$\displaystyle\frac{\mathrm{sd}(\boldsymbol{\eta})}{\sigma}$} 
& \multicolumn{2}{c}{lasso} 
& \multicolumn{2}{c}{GL $\gamma=1$} 
& \multicolumn{2}{c}{GL $\gamma=10$} 
& \multicolumn{2}{c}{GL-select} 
& \multicolumn{2}{c}{ adapt. lasso} 
& \multicolumn{1}{c|}{~} & \it Oracle \\[-0.5ex]
& 
& ~~\scriptsize\it AICc & \multicolumn{1}{c}{\scriptsize\it CV~~}
& ~~\scriptsize\it AICc & \multicolumn{1}{c}{\scriptsize\it CV~~}
& ~~\scriptsize\it AICc & \multicolumn{1}{c}{\scriptsize\it CV~~}
& ~~\scriptsize\it AICc & \multicolumn{1}{c}{\scriptsize\it CV~~}
& ~~\scriptsize\it AICc & \multicolumn{1}{c}{\scriptsize\it CV~~} 
& \multicolumn{1}{c|}{ MCP} & $R^2$ \\[1ex]
\hline
\multicolumn{2}{l|}{\it dense model,} &&&&&&&&&&&\\
\multicolumn{2}{l|}{\it fast decay} &&&&&&&&&&&\\& \it  2  & 11 & 9 & 9 & 7 & 7 & 7 & 7 & 7 & 13 & 12 & {\bf 6} & \it  0.78 \\
\it n=1000  & \it  1  & 9 & 8 & 8 & 7 & 9 & 8 & 8 & 7 & 8 & 8 & {\bf 6} & \it  0.46 \\
& \it  0.5  & {\bf 4} & {\bf 4} & 5 & 5 & 8 & 6 & 5 & {\bf 4} & 6 & 7 & {\bf 4} & \it  0.15 \\[1ex]
\cline{2-2}\rule{0pt}{3ex}& \it  2  & 51 & 46 & 38 & 54 & 13 & 61 & 19 & 47 & 27 & {\bf 10} & 46 & \it  0.68 \\
\it n=100  & \it  1  & 12 & 12 & 12 & 14 & 16 & 15 & 14 & 13 & {\bf 6} & 12 & 12 & \it  0.29 \\
& \it  0.5  & {\bf 0} & {\bf 0} & 14 & {\bf 0} & 26 & 1 & 21 & 1 & 3 & 19 & 1 & \it  0.00 \\[1ex]
\hline\multicolumn{2}{l|}{\it dense model,} &&&&&&&&&&&\\
\multicolumn{2}{l|}{\it slow decay} &&&&&&&&&&&\\& \it  2  & 23 & {\bf 9} & 17 & 10 & 10 & 21 & 10 & {\bf 9} & 20 & 16 & 10 & \it  0.73 \\
\it n=1000  & \it  1  & 12 & 10 & 11 & 13 & 15 & 21 & 11 & 10 & {\bf 9} & {\bf 9} & 10 & \it  0.37 \\
& \it  0.5  & {\bf 2} & 3 & 4 & 3 & 4 & 3 & 3 & 3 & 3 & 4 & 3 & \it  0.07 \\[1ex]
\cline{2-2}\rule{0pt}{3ex}& \it  2  & 47 & 45 & 12 & 54 & {\bf 0} & 56 & 7 & 46 & 23 & 1 & 45 & \it  0.59 \\
\it n=100  & \it  1  & 2 & 3 & 3 & 5 & 5 & 5 & 5 & 4 & {\bf -3} & 1 & 3 & \it  0.09 \\
& \it  0.5  & {\bf -2} & {\bf -2} & 19 & {\bf -2} & 23 & {\bf -2} & 19 & -1 & 1 & 17 & {\bf -2} & \it  -0.06 \\[1ex]
\hline\multicolumn{2}{l|}{\it sparse model,} &&&&&&&&&&&\\
\multicolumn{2}{l|}{\it fast decay} &&&&&&&&&&&\\& \it  2  & 10 & 9 & 8 & 7 & 6 & 6 & 7 & 7 & 12 & 12 & {\bf 5} & \it  0.78 \\
\it n=1000  & \it  1  & 8 & 7 & 7 & 6 & 9 & 7 & 8 & 6 & 7 & 7 & {\bf 5} & \it  0.45 \\
& \it  0.5  & {\bf 3} & {\bf 3} & {\bf 3} & {\bf 3} & 6 & 4 & {\bf 3} & {\bf 3} & 4 & 5 & {\bf 3} & \it  0.12 \\[1ex]
\cline{2-2}\rule{0pt}{3ex}& \it  2  & 49 & 41 & 46 & 40 & 36 & 56 & 38 & 38 & 33 & {\bf 27} & 37 & \it  0.77 \\
\it n=100  & \it  1  & 24 & 24 & 27 & 26 & 32 & 30 & 30 & 24 & {\bf 18} & 26 & 24 & \it  0.44 \\
& \it  0.5  & {\bf 6} & {\bf 6} & 14 & 7 & 33 & 7 & 27 & 7 & 9 & 26 & 7 & \it  0.10 \\[1ex]
\hline\multicolumn{2}{l|}{\it sparse model,} &&&&&&&&&&&\\
\multicolumn{2}{l|}{\it slow decay} &&&&&&&&&&&\\& \it  2  & 14 & 12 & 10 & 9 & 6 & 5 & 6 & 5 & 17 & 17 & {\bf 4} & \it  0.78 \\
\it n=1000  & \it  1  & 14 & 13 & 13 & 13 & 16 & 20 & 14 & 13 & 13 & 13 & {\bf 12} & \it  0.45 \\
& \it  0.5  & {\bf 5} & {\bf 5} & 6 & 6 & 7 & 7 & {\bf 5} & {\bf 5} & 6 & 7 & {\bf 5} & \it  0.12 \\[1ex]
\cline{2-2}\rule{0pt}{3ex}& \it  2  & 52 & 43 & 45 & 43 & 35 & 67 & 39 & 41 & 34 & {\bf 28} & 40 & \it  0.77 \\
\it n=100  & \it  1  & 25 & 25 & 28 & 29 & 33 & 32 & 31 & 25 & {\bf 19} & 27 & 26 & \it  0.44 \\
& \it  0.5  & {\bf 6} & 7 & 18 & 7 & 33 & 7 & 28 & 7 & 9 & 26 & 7 & \it  0.10 \\[1ex]
\hline\end{tabular}
\end{center}
\caption{\label{tab:sumtables} Out-of-sample predictive RMSE, 
reported as  \% worse than oracle (corresponding $R^2$ on far right),
averaged over 1000  samples from various configurations of (\ref{simdgp}).
The dense models have $J=p$ and the sparse models have $J=n/10$.  Fast decay includes $\kappa \in \{10,50\}$, while slow decay is $\kappa \in
\{100,200\}$.
The oracle is MLE fit either on  $C_p$-optimal support for the dense model or
on the true support for the sparse model. Each row of this table corresponds
to average performance across many data generating processes; see the
supplement for more detailed results.  Lasso (GL $\gamma=0$), GL, and AL
routines were executed in {\tt gamlr}.  MCP denotes results from the {\tt
sparsenet} MCP solver. GL-select chooses amongst $\gamma \in \{0,1,10\}$ using
either AICc or CV. The best results are bolded.}
\end{table}

\begin{table}[p]
\footnotesize
\begin{center}
\begin{tabular}{cc|cc|cc|cc|cc|c|c}
& & \multicolumn{9}{l}{\bf Estimation RMSE on true coefficients} & \\[1ex]
& \multirow{2}{*}{$\displaystyle\frac{\mathrm{sd}(\boldsymbol{\eta})}{\sigma}$} 
& \multicolumn{2}{c}{lasso} 
& \multicolumn{2}{c}{GL $\gamma=1$} 
& \multicolumn{2}{c}{GL $\gamma=10$} 
& \multicolumn{2}{c}{ adapt. lasso} 
& \multicolumn{1}{c}{~} & \\[-0.5ex]
& 
& ~~\scriptsize\it AICc & \multicolumn{1}{c}{\scriptsize\it CV~~}
& ~~\scriptsize\it AICc & \multicolumn{1}{c}{\scriptsize\it CV~~}
& ~~\scriptsize\it AICc & \multicolumn{1}{c}{\scriptsize\it CV~~}
& ~~\scriptsize\it AICc & \multicolumn{1}{c}{\scriptsize\it CV~~}
& \multicolumn{1}{c}{ MCP} & \it Oracle \\[1ex]
\hline\multicolumn{2}{l|}{\it sparse model} &&&&&&&&&\\& \it  2  & 0.05 & 0.05 & 0.04 & 0.03 & 0.01 & {\bf 0} & 1.78 & 1.77 & {\bf 0} & \it  0.00 \\
\it n=1000  & \it  1  & 0.12 & 0.12 & {\bf 0.11} & {\bf 0.11} & 0.13 & 0.12 & 3.41 & 3.41 & {\bf 0.11} & \it  0.02 \\
& \it  0.5  & {\bf 0.15} & {\bf 0.15} & {\bf 0.15} & {\bf 0.15} & 0.17 & {\bf 0.15} & 6.47 & 6.54 & {\bf 0.15} & \it  0.10 \\[1ex]
\cline{2-2}\rule{0pt}{3ex}& \it  2  & 0.05 & 0.04 & {\bf 0.03} & 0.04 & 0.04 & 0.06 & 1.07 & 1.02 & 0.04 & \it  0.00 \\
\it n=100  & \it  1  & {\bf 0.07} & 0.08 & 0.08 & 0.08 & 0.12 & 0.08 & 1.89 & 2.02 & 0.08 & \it  0.00 \\
& \it  0.5  & {\bf 0.08} & {\bf 0.08} & 0.12 & {\bf 0.08} & 0.2 & {\bf 0.08} & 3.49 & 4.03 & {\bf 0.08} & \it  0.02 \\[1ex]
\hline\end{tabular}
\end{center}
\vskip -.25cm
\caption{\label{tab:esttables}Summary of estimation RMSE against the 
true coefficients from our {\it sparse} simulation model  (where
 estimation of the true coefficients is a well-posed task). Results are
 averages over 1000 samples from each of the possible data generating process
 configurations, including both binary and continuous designs, $\rho$ in \{0,
 0.5, .9\} and all of our decay values. The oracle is MLE fit on the true
 sparse support and the best results are bolded.}
\end{table}

\begin{figure}[p]
$\boldsymbol{{\kappa} = 10}$

~~~~~~\includegraphics[width=.9\textwidth]{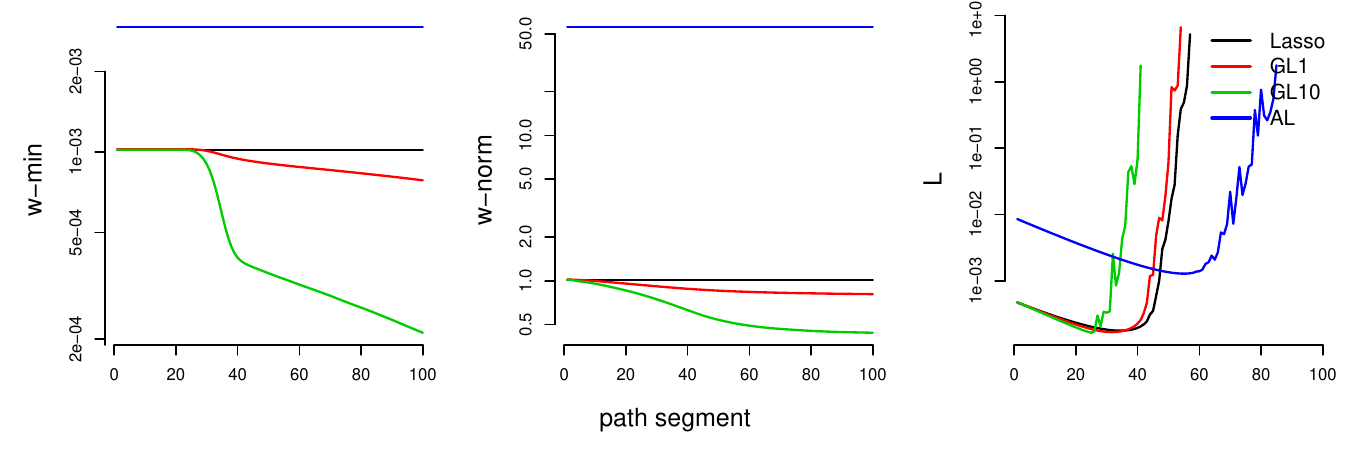}



$\boldsymbol{{\kappa} = 100}$

~~~~~~\includegraphics[width=.9\textwidth]{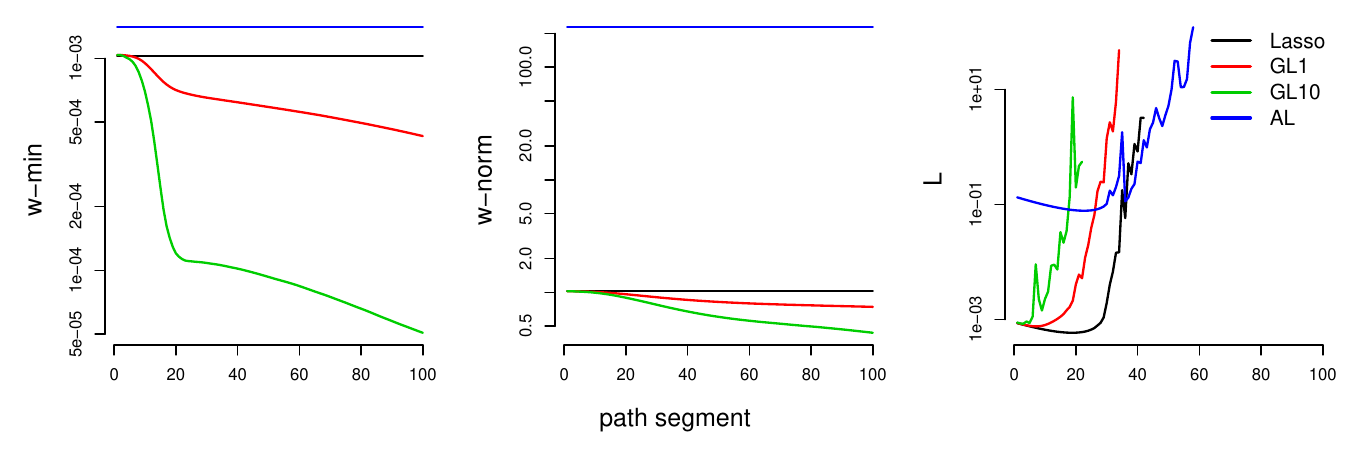}


\caption{\label{fig:weights} Properties of the weighted $\ell_1$ penalized paths in our simulation, averaged across 1000 samples of $n=1000$ for the {\it dense} model with {\it binary} design, $\rho=0.9$, and $\mathrm{sd}(\bs{\eta})/\sigma=1$.
The {\sf\footnotesize w-min} values are $\omega^{\mr{min}}_{S^c}$, {\sf\footnotesize w-norm} are $\|\bs{\omega}_S\|$, and {\sf\footnotesize L} is the restricted eigenvalue constant; all as defined in Theorem \ref{thm:sparseapprox}.  Note that, by construction, the $\lambda$ grids are the same across all algorithms for the same data sample.} 
\end{figure}

\subsubsection*{Remarks}

{\it Predictive performance}

\vskip .2cm \noindent $\bullet$~~ AICc selection for one of the three GL $\gamma$-specifications is
always able to provide predictive performance near to that of the
computationally intensive routine of MCP (with CV selection over a grid in penalty
size and scale).  Looking at the lasso, $\gamma=1$, and $\gamma=10$ columns in
Table \ref{tab:sumtables},  the best OOS RMSE is never more than 2\% worse
than MCP.  Note that MCP is only a dominant performer when the
sample size is large ($n=1000$) and the true model is effectively sparse
(either strictly sparse, or dense with fast decay).

\vskip .2cm \noindent $\bullet$~~   
For $n=1000$, a single {\tt gamlr} run requires a fraction of a second; {\tt
sparsenet} with 5-fold CV usually required 15-20 seconds per run, and
occasionally much longer. The AICc version of GL-select is
within 2\% of the MCP RMSE {\it except} when $n=100$ and signal-to-noise (\textit{s2n}) is 0.5 for the dense model or 0.5-1
for the sparse model.  Thus AICc selection on three GL paths (potentially run
in parallel), with combined compute time still a fraction of a second,  yields
a best or near-best predictor in all scenarios but for these small-data low-signal
settings.  

\vskip .2cm \noindent $\bullet$~~  CV GL-select and MCP provide very similar
performance: their RMSE is always within 1\% of each other, with MCP usually
better at $n=1000$ and GL-select usually better at $n=100$. The combined cost
of each GL CV routine is still far less than a single MCP run, and GL-select
has the advantage that each path -- across different folds and  $\gamma$ --
can be run in parallel.

\vskip .2cm \noindent $\bullet$~~  The relationship between CV and AICc
results for GL is dependent upon $\gamma$, signal strength, and sample size.
For small $n=100$ datasets with strong signal ($s2n=2$), AICc does generally
 better than CV; e.g., this is the one setting where AICc GL-select outperforms MCP by a large margin.  For $n=100$ samples with weak signal ($s2n\leq
1$), the opposite is true and CV outperforms AICc if $\gamma > 0$.  With
larger $n=1000$ samples, AICc and CV perform more similarly {\it except} in
the effectively dense model.  

In summary,  CV and AICc give similar results
whenever $n/s$ is large (regardless of $p$).  CV is safer  when $n/s$ is
small and there is little signal (in this situation even the oracle yields small
or negative $R^2$), while AICc is best if $n/s$ is small but there is strong signal.  The supplement shows that AICc is almost
always a massive improvement over either BIC or  AIC for prediction.

\newpage
\noindent \textit{Estimation performance and information compression}

\vskip .2cm \noindent $\bullet$~~ Table \ref{tab:esttables} shows mostly
similar estimation errors between the GL and MCP comparators. MCP dominates only in the $n=1000$ samples with strong $s2n=2$ signal.  The simple lasso is also competitive, despite its non-diminishing bias, in all but the larger samples with $s2n \geq1$.

\vskip .2cm \noindent $\bullet$~~ Marginal AL performed well in prediction,
with RMSE that was near to that of the lasso when $n=1000$ and
sometimes much better than MCP or GL when $n=100$.  However, Table \ref{tab:esttables}
shows {\it terrible} performance in terms of estimation error for marginal AL.
It is interesting that marginal AL outperforms MCP in prediction when $n=100$ and
$s2n\geq1$, but has estimation RMSE that is an order of magnitude larger in the
same scenario.

\vskip .2cm \noindent $\bullet$~~  In the supplement, CV and AICc selected GL1 has usually 20-80\% the number of nonzero
coefficients as the corresponding  lasso fit.  GL10 leads to
even more sparsity, often returning less than 10\% of the selected {\it df} from lasso. This suggest  {\tt gamlr}
with a small but nonzero $\gamma$ (e.g., as in GL1) as a reliable strategy for
compressing information without hurting predictive performance (our original motivating goal).  The supplement also
shows that $\gamma>0$ leads to a large drop in false discovery (with respect to oracle support) relative
the the lasso.   
In contrast,  marginal AL yields no more (and often less) sparsity than the standard
lasso and provides none of the desired information compression gain.  
Overall, MCP seems to have best variable selection properties (reducing FDR without dramatically lower sensitivity).

\vskip .4cm
\noindent \textit{Path properties}

\vskip .2cm \noindent $\bullet$~~  
 Figure \ref{fig:weights}
shows how the quantities in Theorem \ref{thm:sparseapprox} behave  for various coefficient
decay rates in a highly collinear ($\rho=.9$) and moderately noisy
($s2n=1$) setting.  Recall that prediction error
should get closer to that of an $\ell_0$ oracle if $\|\bs{\omega}_S\|$ -- the
norm on the weights in the support of the $\ell_0$ rule -- can be made small
without $\omega^{\mr{min}}_{S^c}$ -- the smallest weight on the complement of
this support -- shrinking too much.    In comparison to lasso, this is achieved for $\gamma=1$ under fast coefficient decay ({$\kappa=10$}, high effective sparsity) while 
$\omega^{\mr{min}}_{S^c}$  drops right after the path begins under slower decay and lower effective sparsity.   For the fast
diminishing bias of $\gamma=10$,  $\omega^{\mr{min}}_{S^c}$ drops more dramatically
and earlier; only in the ${\kappa}=10$ setting does there
appear to be any opportunity for improved prediction from $\gamma=10$ relative
to $\gamma=1$.  Marginal AL provides an interesting side-case; it maintains higher
$\omega^{\mr{min}}_{S^c}$ at the expense of a {\it much} higher
$\|\bs{\omega}_S\|$. 

\vskip .2cm \noindent $\bullet$~~ The far right panels in Figure
\ref{fig:weights} show the realized value of $L = 
\frac{\|\bs{\omega}_S\|}{\sqrt{s}}\left(\omega^{\mr{min}}_{S^c}-\sqrt{2}\sigma/\lambda\right)^{-1}$, the
constant governing our restricted eigenvalue $\phi^2(L,S)$ at the $C_p$-optimal 
model (see Definition \ref{redef}).  A small value of this constant is
important for keeping prediction error low (by keeping $\phi^2(L,S)$ big).
Moreover, in the supplement we show that small values of $L$ are essential in
controlling the false discovery rate.  In Figure \ref{fig:weights}, we see
that $L$ is decreasing steadily for each algorithm until the point where
$\omega^{\mr{min}}_{S^c}$ drops; after this point it explodes and, soon after
that, becomes undefined when $\omega^{\mr{min}}_{S^c} <
\sqrt{2}\sigma/\lambda$.

\section{Hockey example}
\label{sec:nhl}\vskip -.1cm

We close with an example analysis: measuring the
performance of hockey players.  It extends  analysis in
\cite{gramacy_estimating_2013,gramacy_hockey_2015}.  The  data include every goal in the National Hockey League (NHL) back to the
2002-2003 season: 69,449 goals and 2439 players.

For goal $i$ in
season $s$ with away team $a$ and home team $h$, say that $q_{i}$ is the probability that the home team scored this goal.  Our regression model is then
\begin{align}\label{hockeymod}
\mr{logit}\left[q_{i}\right] = \alpha_0 +
\alpha_{sh} - \alpha_{sa} + \bm{u}_i'\bs{\phi} + \bm{x}_i'\bs{\beta}_0 +
\bm{x}_i'\bs{\beta}_s, \end{align}  Vector $\bm{u}_i$ holds indicators for
various special-teams scenarios (e.g., a home team power play), and
$\bs{\alpha}$ provides matchup/season specific intercepts. Vector $\bm{x}_i$
contains player effects: $x_{ij}=1$ if player $j$ was on the home team and on
ice for goal $i$, $x_{ij}=-1$ for away player $j$ on ice for goal $i$, and
$x_{ij}=0$ for everyone not on the ice.   Coefficient $\beta_{0j} +
\beta_{sj}$ is the season-$s$ effect of player $j$ on the log odds that, given
a goal has been scored, the goal was scored by their team.  These effects are
`partial' in that they control for who else was on the ice, special teams
scenarios, and team-season fixed effects -- a player's $\beta_{0j}$ or $\beta_{sj}$
only need be nonzero if that player effects play above or below the team
average for a given season.

We estimate GL paths of $\bs{\hat\beta}$ from (\ref{hockeymod}) with
$\bs{\alpha}$ and $\bs{\phi}$ left {\it unpenalized}. Coefficient costs are
{\it not} scaled by covariate standard deviation, since this would have
favored players with little ice time.   Joint $[\gamma,\lambda]$ surfaces for AICc and BIC are in Figure \ref{nhlic}.
AICc favors denser models with low $\lambda$ but not-to-big $\gamma$,
while the BIC  prefers very sparse but relatively unbiased  models with large
$\lambda$ and small $\gamma$.  Both criteria are strongly adverse to any model
above $\gamma=100$, which is also where timings explode (supplement).  Ten-fold CV results are shown in Figure
\ref{nhlcv} for $\gamma$ of 0, 1, and 10.  The OOS error minima are around the
same in each case -- average deviance slightly above 1.16 -- but errors
increase much faster away from optimality with larger $\gamma$.  AICc selection is always between the CV error-minimizing selection and that of the common {\it 1SE} rule (largest $\lambda$ with mean OOS error no more than
1 standard error away from the minimum).

\begin{figure}[tb]
\begin{center}\vskip -.5cm
~~~~~~~~\includegraphics[width=5.5in]{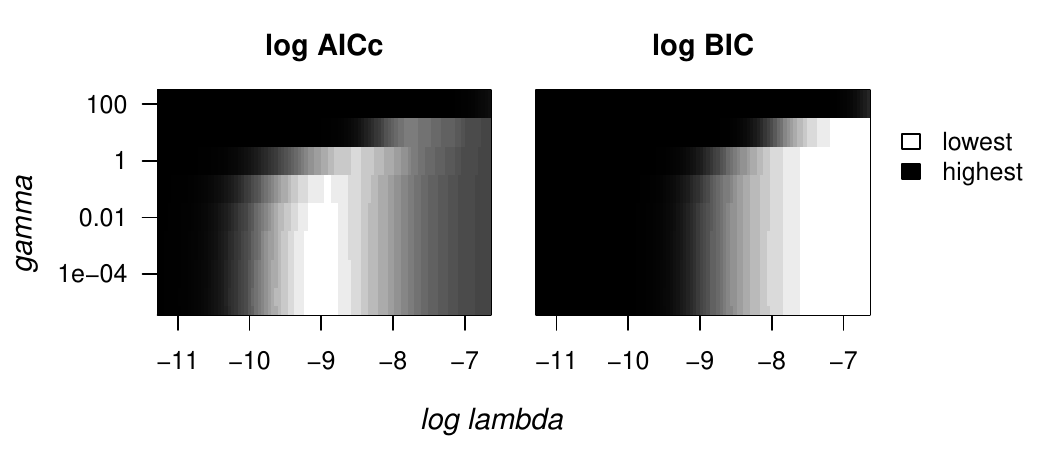}
\caption{ \label{nhlic} Hockey example AICc and BIC surfaces, rising from white to black on log scale.
}
\end{center}
\includegraphics[width=6.3in]{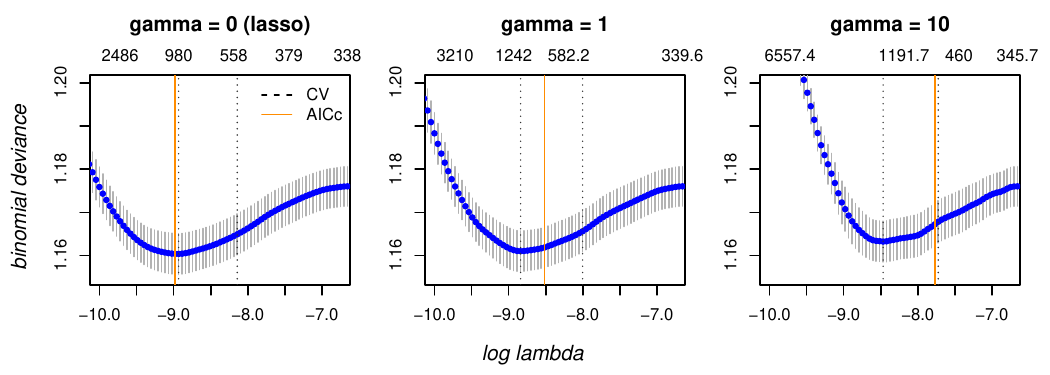}
\caption{\label{nhlcv} Hockey example 10-fold CV: mean OOS deviance $\pm 1$se, with minimum-error and 1SE selection rules marked with black dotted lines, and solid orange line showing AICc selection. }
\end{figure}

\begin{table}[tb]\footnotesize\setstretch{1}
\makebox[\linewidth]{
\begin{tabular}{l|lcc|lcc|lcc|}
\multicolumn{1}{c}{} & \multicolumn{3}{l}{\normalsize \it ~~~lasso} & \multicolumn{3}{l}{\normalsize \it ~~~$\gamma=1$} &  \multicolumn{3}{l}{\normalsize \it ~~~$\gamma=10$}\\
\multicolumn{1}{c}{}  &     &    PPM & PM &    &   PPM & PM        &   & PPM & PM \\
\cline{3-4}\cline{6-7}\cline{9-10}
1 & Ondrej Palat & 33.8 & 38 & Sidney Crosby & 29.2 & 52 & Sidney Crosby & 32.6 & 52 \\
2 & Sidney Crosby & 31.2 & 52 & Ondrej Palat & 29 & 38 & Jonathan Toews & 22.8 & 35 \\
3 & Henrik Lundqvist & 25.8 & 9 & Jonathan Toews & 21.4 & 35 & Joe Thornton & 22 & 34 \\
4 & Jonathan Toews & 24 & 35 & Joe Thornton & 21 & 34 & Anze Kopitar & 22 & 39 \\
5 & Andrei Markov & 23.1 & 34 & Andrei Markov & 20.9 & 34 & Andrei Markov & 20.7 & 34 \\
6 & Joe Thornton & 21.4 & 34 & Henrik Lundqvist & 19.8 & 9 & Alex Ovechkin & 18.1 & 16 \\
7 & Anze Kopitar & 20.6 & 39 & Anze Kopitar & 19.5 & 39 & Pavel Datsyuk & 16.6 & 13 \\
8 & Tyler Toffoli & 18.9 & 31 & Pavel Datsyuk & 16.1 & 13 & Ryan Getzlaf & 15.8 & 16 \\
9 & Pavel Datsyuk & 17.7 & 13 & Logan Couture & 15.9 & 29 & Henrik Sedin & 15.2 & 7 \\
10 & Ryan Nugent-hopkins & 17.4 & 18 & Alex Ovechkin & 15.8 & 16 & Marian Hossa & 14.9 & 21 \\
11 & Gabriel Landeskog & 16.6 & 36 & Marian Hossa & 14.4 & 21 & Alexander Semin & 14.7 & -1 \\
12 & Logan Couture & 16.5 & 29 & Alexander Semin & 14.2 & -1 & Jaromir Jagr & 14.5 & 28 \\
13 & Alex Ovechkin & 15.8 & 16 & Matt Moulson & 13.9 & 22 & Logan Couture & 14.2 & 29 \\
14 & Marian Hossa & 15.4 & 21 & Tyler Toffoli & 13.3 & 31 & Matt Moulson & 13.7 & 22 \\
15 & Alexander Semin & 14.8 & -1 & David Perron & 12.7 & 2 & Mikko Koivu & 13 & 12 \\
16 & Zach Parise & 14.7 & 21 & Mikko Koivu & 12.5 & 12 & Joe Pavelski & 12.6 & 33 \\
17 & Frans Nielsen & 13.5 & 8 & Frans Nielsen & 12.3 & 8 & Steven Stamkos & 12.6 & 24 \\
18 & Mikko Koivu & 13.4 & 12 & Ryan Getzlaf & 12.1 & 16 & Frans Nielsen & 12.5 & 8 \\
19 & Matt Moulson & 13.4 & 22 & Ryan Nugent-hopkins & 11.9 & 18 & Marian Gaborik & 12.3 & 29 \\
20 & David Perron & 13.1 & 2 & Jaromir Jagr & 11.8 & 28 & Zach Parise & 12.2 & 21 \\
\multicolumn{1}{c}{} & \multicolumn{3}{r|}{} & \multicolumn{3}{r|}{} &  \multicolumn{3}{r|}{}\\
 \multicolumn{1}{c}{} & \multicolumn{3}{r|}{\it 305 nonzero effects} & \multicolumn{3}{r|}{\it 204 nonzero effects} &  \multicolumn{3}{r|}{\it 64 nonzero effects}
\end{tabular}}
\caption{\label{nhleffects} Top 20 AICc selected player `partial plus-minus' (PPM) values for the 2013-2014 season, under $\gamma = 0,1,10$.  The number of nonzero player effects for each $\gamma$ are noted along the bottom.}
\end{table}

The original goal with this dataset was to build a better version of
hockey's `plus-minus' (PM) statistic: number of goals {\it for} minus {\it
against} each player's team while he is on the ice. To convert from player
effects $\beta_{0j} + \beta_{sj}$ to the scale of `plus/minus', set the probability that a goal was scored by his team given player
$j$ is on ice  (and no other information) as $p_j = e^{\beta_j}/(1+e^{\beta_j})$. The `partial
plus/minus' (PPM) is 
\begin{equation}
 \mr{ppm}_j = N_j(p_j - (1-p_j)) = N_j(2p_j-1)
 \end{equation}  where
$N_j$ is the  number of goals for which he was on-ice.  This measures 
quality and quantity of contribution and lives on the same scale as PM.  See \cite{gramacy_hockey_2015} for details.

Table \ref{nhleffects} contains the estimated PPM values for the 2013-2014
season under various $\gamma$ levels, using AICc selection.  We see that, even
if changing concavity ($\gamma$) has little effect on minimum CV errors (Figure \ref{nhlcv}),
larger $\gamma$  yield more sparse models and different conclusions about
player contribution. At the $\gamma=0$ lasso, there are 305 nonzero player
effects (individuals measurably different from their team's average ability)
and the list includes young players who have had very strong starts to their
careers.  For example, Ondrej Palat and Tyler Toffoli both played their first
full seasons in the NHL in 2013-2014.  As $\gamma$ increases to 1, these young
guys  drop in rank while more proven stars (e.g., Sidney Crosby and Jonathan
Toews) move up the list.  Finally, at $\gamma=10$ only big-name stars remain
amongst the 64 nonzero player effects.

\section{Discussion}
\label{discussion}

Whenever exact solvers are too
 expensive, concave penalized estimation reduces largely to weighted-$\ell_1$ penalization.
Path adaptation is an intuitively reasonable source of weights, and we are able to show that POSE -- particularly {\tt gamlr} with AICc selection -- provides high quality diminishing-bias sparse regression at
{\it no more cost} than a  single lasso path.  We know of no other software that meets this standard.

\setstretch{1}
\bibliographystyle{chicago}
\bibliography{taddy}

\end{document}